\documentclass[aps,pra,reprint,superscriptaddress]{revtex4-2}

\usepackage{adjustbox}
\usepackage{amsmath}
\usepackage{amssymb}
\usepackage{amsthm}
\usepackage{bbm}
\usepackage{color}
\usepackage[T1]{fontenc}
\usepackage{graphicx}
\usepackage{hyperref}
\usepackage[utf8]{inputenc}
\usepackage{mathtools}
\usepackage{natbib}
\usepackage{orcidlink}
\usepackage{physics}
\usepackage{times}
\usepackage{tikz}
\usepackage{xcolor}

\newtheorem{condition}{Condition}[section]
\newtheorem{theorem}{Theorem}[section]
\newtheorem{axiom}{Axiom}[section]

\newtheorem{corollary}{Corollary}[section]
\newtheorem{definition}{Definition}[section]
\newtheorem{example}{Example}[section]

\newtheorem{proposition}{Proposition}[section]
\newtheorem{remark}{Remark}[section]

\newcommand{\CZ}{\mathrm{CZ}}

\newcommand{\CCZ}[1]{\mathrm{C}^{#1}\mathrm{Z}}

\DeclareMathOperator{\diag}{diag}

\hypersetup{
    colorlinks=true, linkcolor=blue, citecolor=blue,
    filecolor=blue, urlcolor=blue, breaklinks=true
}

\newcommand{\qpqi}{
  QPQI Group,
  Universidade Estadual de Ponta Grossa,
  84030-900 Ponta Grossa, Paraná, Brazil
}

\newcommand{\demat}{
  Departamento de Matem\'{a}tica e Estat\'{i}stica,
  Universidade Estadual de Ponta Grossa,
  84030-900 Ponta Grossa, Paraná, Brazil
}

\newcommand{\orcidnathan}{\orcidlink{0009-0007-0666-2266}}
\newcommand{\orcidalison}{\orcidlink{0000-0003-3552-8780}}
\newcommand{\orcidgiuliano}{\orcidlink{0000-0002-0205-9597}}
\newcommand{\orcidfabiano}{\orcidlink{0000-0001-5383-6168}}

\begin{document}

\title{Encoding matroids into quantum states}

\author{Nathan Ferreira\orcidnathan}
\email{nnferreira98@gmail.com}
\affiliation{\qpqi}

\author{Alison A. Silva\orcidalison}
\email{alisonantunessilva@gmail.com}
\affiliation{\qpqi}

\author{Giuliano G. La Guardia\orcidgiuliano}
\email{gguardia@uepg.br}
\affiliation{\demat}

\author{Fabiano M. Andrade\orcidfabiano}
\email{fmandrade@uepg.br}
\affiliation{\qpqi}
\affiliation{\demat}

\date{\today}

\begin{abstract}
Efficient representations of multipartite quantum states play a fundamental role in quantum information theory, providing both conceptual insight and practical tools for characterizing entanglement. Motivated by the axiomatic framework for graph states \href{https://doi.org/10.1103/PhysRevA.85.062313}{[Phys. Rev. A \textbf{85}, 062313 (2012)]} and its subsequent extension to hypergraph states \href{https://doi.org/10.1103/PhysRevA.87.022311}{[Phys. Rev. A \textbf{87}, 022311 (2013)]}, we introduce an axiomatic construction of \emph{matroid states}, a new family of multipartite quantum states associated with matroids. Our constructions are based on a set of axioms analogous to those that define graph and hypergraph states, yielding a consistent quantum representation of arbitrary matroids. Two ways of constructing matroid states are proposed: the first is defined in terms of circuits, and the second in terms of independent sets. In both approaches, we establish the existence of universal global operators that satisfy desirable properties such as locality, symmetry, commutativity, and are associated with the combinatorial structure of matroids. Furthermore, we establish a hierarchy connecting graph, matroid, and hypergraph states within a unified framework. Additionally, we show how to obtain an arbitrary graph state by applying suitable families of matroid states, whose corresponding operators are the generators of the stabilizer subgroup of the graph state. These results identify matroid theory as a natural combinatorial language for the efficient description of multipartite quantum states and open new perspectives for the investigation of quantum entanglement and related combinatorial structures.
\end{abstract}

\maketitle

\section{Introduction}
\label{sec:intro}

The efficient description of multipartite quantum states remains one of the central challenges in quantum information theory \cite{Marconi:2026}. Although the Hilbert space associated with a multipartite quantum system grows exponentially with the number of constituents \cite{Nielsen:2010,vonNeumann:1955}, physically relevant families of states often possess an underlying mathematical structure that enables a significantly more compact representation. Such representations not only facilitate the characterization of multipartite entanglement but also establish deep connections between quantum information and areas of discrete mathematics, leading to new theoretical insights and computational techniques.

The idea of encoding quantum states in mathematical structures that are not Hilbert spaces seeks to capture the relational features of complex vector spaces associated with different quantum subsystems that comprise the subject of analysis. Among the most successful examples of this interplay are graph states, which associate multipartite quantum states with simple graphs \cite{Hein:2004,Aschauer:2005}. Besides their fundamental role as universal resources for measurement-based quantum computation (MBQC) \cite{Raussendorf:2001,Raussendorf:2003,Briegel:2009}, graph states provide an elegant combinatorial description of entanglement in which many physical properties can be inferred directly from graph-theoretical concepts. Seeking a unifying principle behind these constructions, Ionicioiu and Spiller  \cite{Ionicioiu:2012} introduced an axiomatic framework for encoding graphs into quantum states. Rather than defining graph states through a specific quantum circuit, their approach is based on three physically motivated axioms: separability, graph isomorphism, and the existence of a universal edge operator, from which several important families of multipartite states naturally emerge, including graph states, qudit graph states, Gaussian cluster states, projected entangled pair states, and quantum random networks. This framework established a general methodology for associating combinatorial structures with multipartite quantum systems.

This axiomatic construction was subsequently generalized from graphs to hypergraphs by Qu \emph{et al.} \cite{Qu:2013a}. By modifying the original axioms and introducing universal hyperedge operators acting on arbitrary subsets of vertices, they defined the family of hypergraph states and showed that these states considerably enlarge the class of multipartite entangled states described within the axiomatic formalism. Their work further clarified the relationship between graph states, hypergraph states, and stabilizer states, illustrating the versatility of the axiomatic approach.

Despite these developments and the fact that they constitute an important class of multipartite entangled states \cite{Hein:2006,Qu:2013b,Qu:2013c,Rossi:2013,Guhne:2014,Poderini:2026}, both graph and hypergraph states remain fundamentally based on incidence structures. However, many of the combinatorial concepts that naturally arise in graph theory and are frequently relevant in quantum information are more fundamentally described by the notion of independence rather than incidence. Indeed, concepts such as circuits, rank, bases, cuts, and cycle spaces admit a unified treatment within matroid theory \cite{Oxley:1992}. Introduced by Whitney \cite{Whitney:1935} as an abstraction of linear independence, matroids encompass a wide variety of combinatorial structures, including graphs, vector spaces, linear codes, and optimization problems.

In this work, motivated by these considerations, we propose an axiomatic framework for encoding matroids into multipartite quantum states. Our objective is not merely to extend the graph and hypergraph constructions to another combinatorial object, but rather to investigate whether the notion of matroid provides a more fundamental setting for describing multipartite quantum states while preserving the physical principles underlying the original axiomatic framework. To this end, we formulate a set of axioms analogous to those proposed by Ionicioiu and Spiller and derive the corresponding quantum representation associated with arbitrary matroids. More precisely, we construct matroid states from matroids in two ways: the first is defined in terms of circuits, and the second in terms of independent sets. Within this framework, we establish the existence of universal global operators that satisfy desirable properties such as locality, symmetry, and commutativity, and are associated with the combinatorial structure of matroids. We show that the resulting construction is mathematically consistent. We also demonstrate that every matroid state belongs to the class of hypergraph states, thereby placing graph, matroid, and hypergraph states within a single hierarchical framework. These results identify matroid theory as a natural combinatorial language for describing multipartite quantum states and provide a new connection between matroid theory and quantum information, opening avenues for investigating multipartite entanglement and related combinatorial quantum structures.

The paper is organized as follows. In Section~\ref{sec:preli}, we recall some concepts and results concerning matroid theory that will be utilized throughout this paper. In Section~\ref{sec:mscirc}, we propose a new method for generating a quantum state using minimal dependent sets (circuits) of a matroid. In Section~\ref{sec:msstab}, we show that a graph state can be constructed from $t+1$ matroid states, where $t$ is the number of incident vertices of the corresponding graph. Section~\ref{sec:mssets} presents constructions of matroid states by means of independent sets. In Section~\ref{sec:disc}, we discuss the results presented here, and in Section~\ref{sec:fr}, we draw the final comments.

\section{Matroid Preliminaries}
\label{sec:preli}

In his seminal paper introducing matroid theory \cite{Whitney:1935}, Whitney sought to capture properties common to vector spaces, graph theory, and algebraic dependence, among others.
In this section, we recall some definitions and results on matroid theory that will be necessary for the development of this work. For more details on matroids, we refer the reader to \cite{Oxley:1992}. 

\begin{definition}\label{matro}
A matroid $M$ is an ordered pair $(X, {\mathcal I})$ consisting of a
finite nonempty set $X$ and a collection ${\mathcal I}$ of subsets of $X$ 
satisfying the following conditions:
\begin{itemize}
\item [ ${\bf (I1)}$] \ $\emptyset \in {\mathcal I}$. 

\item [ ${ \bf (I2)}$] If $I \in {\mathcal I}$ and $I^{\prime} \subseteq I$, then
$I^{\prime} \in {\mathcal I}$. 

\item [ ${ \bf (I3)}$] If $I_1, \ I_2 \in {\mathcal I}$ and $ | I_{1}|< |I_{2}|$, then
there exists an element $e \in I_{2}\setminus I_{1}$ such that $I_{1} \cup \{
	e \} \in {\mathcal I}$, where $|.|$ denotes the cardinality of the set.
\end{itemize}
\end{definition}

If $M$ is the matroid $(X, {\mathcal I})$, written $M=(X, {\mathcal I})$, then $M$ is called \emph{matroid on} $X$. The members of ${\mathcal I}$ are called
\emph{independent sets} of $M$, and $X$ is the \emph{ground set} of
$M$. A subset of $X$ that is not in ${\mathcal I}$ is said to be a \emph{dependent set}. Minimal dependent sets, i.e., \emph{circuits}, are
dependent sets all of whose proper subsets are independent sets.
A subset of $X$ is called a {\em basis} (or {\em base}) of $M$ if it is a maximal independent set. The set of all bases of $M$ is denoted by ${\mathcal B}$. All bases of a matroid $M$ have the same cardinality; the \emph{rank} of a matroid $M$ is the cardinality of a basis of $M$.

An important and well-known class of matroids is the class of \emph{vector matroids} of a given matrix $A$ over some field ${\mathbb F}$. 
\begin{theorem}
Let $S$ be a set of column labels of a matrix $A_{m\times n}$ over a
field $\mathbb{F}$. Let ${\mathcal I}$ be the collection of subsets $Y$ of
$S$ for which the multiset of columns labeled by $Y$ is linearly
independent (LI) in $V(m, \mathbb{F})$, the $m$-dimensional vector
space over ${\mathbb F}$. Then $(S, {\mathcal I})$ is a matroid.
\end{theorem}
\begin{proposition}\label{propcircuits}
    The set $\mathcal{C}$ of circuits of a matroid $M$ satisfies the following properties:
\begin{itemize}
        \item [ {\bf (C1)}] $\emptyset \notin {\mathcal C}$.
        \item [ { \bf (C2)}] If $C_1, C_2 \in {\mathcal C}$ and $C_1\subseteq C_2$ then $C_1=C_2$.
        \item[ { \bf (C3)}] If $C_1$ and $C_2$ are distinct members of ${\mathcal C}$ and $e \in C_1 \cap C_2$, then there exists a member $C_3$ of ${\mathcal C}$ such that $C_3 \subseteq (C_1 \cup C_2) \setminus \{e\}$.
\end{itemize}    
\end{proposition}

A matroid can also be characterized by its set of circuits, according to the following result.

\begin{theorem}\label{caracircuits}
    Let $X$ be a nonempty set and $\mathcal{C}$ be a collection of subsets of $X$ that satisfy $\bf(C1)$, $\bf(C2)$, and $\bf(C3)$. Let ${\mathcal I}$ be the collection of subsets of $X$ that do not contain members of $\mathcal{C}$. Then $(X, {\mathcal I})$ is a matroid that has ${\mathcal C}$ as its collection of circuits.
\end{theorem}

The following result shows how to obtain a matroid $M_G$ derived from a graph $G$, if the ground set is the set of edges of $G$. This matroid $M_G$ is called \emph{cycle matroid} of $G$.

\begin{theorem}\label{cyclemat}
    If $E$ is the set of edges of a graph $G$, and ${\mathcal C}$ is the set of edges of cycles of $G$, then ${\mathcal C}$ is the set of circuits of a matroid on $E$.
\end{theorem}

A matroid that is isomorphic to a cycle matroid is called \emph{graphic}, where \emph{isomorphism} of matroids is defined in the sequence.

\begin{definition}
\label{isomat}
    Let $M_1=(X_1 , {\mathcal I}_1 )$ and $M_2=(X_2 , {\mathcal I}_2 )$ be two matroids. One says that $M_1$ and $M_2$ are isomorphic, written $M_1 \cong M_2$, if there exists a bijection $P_{1,2}: X_1 \longrightarrow X_2$ that preserves independent sets, i.e., if for all $A \subseteq X_1$, $P_{1,2}(A)$ is independent in $M_2$ if and only if $A$ is independent in $M_1$.
\end{definition}

Definition~\ref{isomat} can be rewritten in terms of circuits.

\begin{definition}\label{isomatrew}
    Let $M_1=(X_1 , {\mathcal C}_1 )$ and $M_2=(X_2 , {\mathcal C}_2 )$ be two matroids. One says that $M_1$ and $M_2$ are isomorphic, written $M_1 \cong M_2$, if there exists a bijection $P_{1,2}: X_1 \longrightarrow X_2$ that preserves circuits, i.e., if for all $A \subseteq X_1$, $P_{1,2}(A)$ is a circuit in $M_2$ if and only if $A$ is a circuit in $M_1$.
\end{definition}

\section{Matroid States from Circuits}
\label{sec:mscirc}

In this section, we introduce the concept of \emph{matroid quantum states}, or \emph{matroid states} for short. As is well known, the theory of graph states captures the properties of a given graph and transfers them to the corresponding quantum state. In general, given the underlying graph, one assigns a unique operator to each edge, which, in several cases, is the controlled-$Z$ gate, $\CZ \equiv \diag(1,1,1,-1)$. The procedure continues iteratively, edge by edge, until the final state is achieved \cite{Ionicioiu:2012}. 

In this section, we denote a matroid $M= (X, {\mathcal C})$ in terms of its collection of circuits ${ \mathcal C}$. We utilize the fundamental structures of the circuits of $M$ to define its corresponding matroid state $\ket{M}$. In other words, we assign to each circuit $C$ of $M$ a unique operator $U_C$ that depends uniquely on $C$ (see Axiom~\ref{axiom3c}). 

As usual, the first concept to be considered here is the \emph{direct sum} of matroids.
Let $X_1$ and $X_2$ be disjoint sets, and $M_1 =(X_1,{ \mathcal C}_1)$ and $M_2 =(X_2,{ \mathcal C}_2)$ be two matroids defined in terms of their corresponding circuit sets. The direct sum $M_1 \oplus M_2$ is the matroid with ground set $X_1 \cup X_2$ and circuit set ${\mathcal C}_{M_1 \oplus M_2}= {\mathcal C}_1 \cup {\mathcal C}_2$. 
Because the collections of circuits ${\mathcal C}_1$ and ${\mathcal C}_2$ of $M_1$ and $M_2$, respectively, are disjoint, it makes sense to define the separability of matroid states by direct sum.

\begin{axiom}\label{axiom1c}
    \emph{Separability}. Let $M_1=(X_1 , {\mathcal C}_1 )$ and $M_2=(X_2 , {\mathcal C}_2 )$ be two matroids on disjoint sets $X_1$ and $X_2$. Then one has $\ket{M_1 \oplus M_2} = \ket{M_1} \otimes \ket{M_2}$.
\end{axiom}

\begin{definition}\label{empcircuit}
    Let $M_{n}^{\emptyset} =(X, {\mathcal C})$ be a matroid with the collection of circuits ${\mathcal C}= \emptyset$ and $|X|=n$. Then $M_{n}^{\emptyset}$ is called the empty matroid on $n$ elements.
\end{definition}

\begin{corollary}\label{coro1c}
    If $M_{n}^{\emptyset}=(X , \emptyset )$ is the empty matroid, where $|X|=n$, then 
    $\ket{M_{n}^{\emptyset}}= \ket{{\phi}_1} \otimes \ket{{\phi}_2} \otimes \ldots \otimes \ket{{\phi}_n}$ is a separable state.
\end{corollary}
\begin{proof}
    Since we associate an operator with each circuit, if there is no circuit, no operator acts on the quantum state.
\end{proof}

\begin{corollary}\label{coro2c}
    Given a matroid $M=(X, { \mathcal C})$ with $|X|=n$, we associate with each element $i \in X$ a Hilbert space ${\mathcal H}_i$. The total Hilbert space is given by ${\mathcal H} = \displaystyle \bigotimes_{i=1}^{n}{\mathcal H}_i$.
\end{corollary}

\begin{axiom}\label{axiom2c}
    \emph{Matroid Isomorphism}. Let $M_1=(X_1 , {\mathcal C}_1 )$ and $M_2=(X_2 , {\mathcal C}_2 )$ be two isomorphic matroids. Then the corresponding density operators ${\rho}_1 = \ketbra{M_1}$ and ${\rho}_2 = \ketbra{M_2}$ satisfy
\begin{eqnarray}
    {\rho}_2 = D( P_{1,2}) {\rho}_1 [D( P_{1,2})]^{-1},
\end{eqnarray}
where $D( P_{1,2})$ is a matrix representation of the bijection $P_{1,2}$ mapping $M_1$ to $M_2$.
\end{axiom}

\begin{corollary}\label{coro3c}
    If $P_M$ is an automorphism of the matroid $M=(X, {\mathcal C})$, then $\left[ \rho , D(P_{M} )\right] =0$.
\end{corollary}
\begin{proof}
    Follows directly from Axiom~\ref{axiom2c}.
\end{proof}

\begin{proposition}\label{prop1c}
    Let $M=(X, {\mathcal C})$ be a matroid, where $|X|=n$.
    Then the corresponding matroid state $\ket{M}$ belongs to a Hilbert space ${ \mathcal H}$ of $n$ identical quantum systems ${\mathcal H}_1$, i.e., ${ \mathcal H}=  {\mathcal H}_1^{\otimes n}$, where ${\mathcal H}_1$ is a Hilbert space associated with a single element of the ground set $X$. Additionally, the empty matroid $M_{n}^{\emptyset}$ is assigned to $M_{n}^{\emptyset} \rightarrow \ket{M_{n}^{\emptyset}} ={\ket{\phi}}^{\otimes n}$, where $\ket{\phi} \in \mathcal {H}_1$.
\end{proposition}
\begin{proof}
    The proof is the same as that given in Proposition 1 of \cite{Ionicioiu:2012}.
\end{proof}

It is interesting to note that the dimension of ${\mathcal H}_1$ is arbitrary, i.e., it is a free parameter of the theory. Another consequence of Proposition~\ref{prop1c} is that all matroids whose ground set has $n$ elements are mapped to the same Hilbert space ${ \mathcal H}=  {\mathcal H}_1^{\otimes n}$. Hence, fixing the ground set $X$, given two matroids $M_1=(X, {\mathcal C}_1)$ and $M_2=(X, {\mathcal C}_2)$ with $|X|=n$, there exists a linear operator $U_{M_1 , M_2}$ on ${ \mathcal H}$ such that $\ket{M_2} = U_{M_1 , M_2} \ket{M_1}$. In particular, it follows that 
\begin{equation}\label{empty}
\ket{M} = U_{M_{n}^{\emptyset} , M}\ket{M_{n}^{\emptyset}}=U_{M_{n}^{\emptyset} , M}{\ket{\phi}}^{\otimes n}.
\end{equation}

Before proceeding further, let us recall the following axiom shown in \cite{Ionicioiu:2012}

\begin{axiom}\cite[Axiom A3]{Ionicioiu:2012}\label{axiomspiller}
    Universal edge operator. If the graphs $G=(V, E)$ and $G'=(V',E')$ differ by a single edge, i.e., $V=V'$ and $E'=E\cup \{x, y\}$, then $\ket{G'} = U_{\{x,y\}} \ket{G}$. The edge operator $U_{\{x,y\}}$ is independent of both $G$ and $G'$ and depends uniquely on the edge $\{x,y\}$.
\end{axiom}

In the case of a matroid, it can be necessary to add only one circuit or even a finite number of circuits due to Item ${ \bf (C3)}$ of Proposition~\ref{propcircuits}, as we will see in the next two examples. In Example~\ref{exe1}, it is necessary to add only one circuit to ${\mathcal C}_1$, and in Example~\ref{exe2}, it is necessary to add two circuits to ${\mathcal I}_1$ at the same step.

\begin{example}\label{exe1}
    Let $M_1 = (X, {\mathcal C}_1)$ be a matroid with $X=\{ a, b, c, d \}$, ${\mathcal C}_1 = \{ \{a, b\} \}$. Let $M_2 = (X, {\mathcal C}_2)$ be another matroid with ${\mathcal C}_2 = \{ \{a, b\}, \{c, d\}\}$. Then, there exists a unique operator $U_{\{c, d\}}$, corresponding to the circuit $\{c, d\}$ such that $\ket{M_2} = U_{\{b, c\}}\ket{M_1}$.
\end{example}

\begin{example}\label{exe2}
    Let $M_1 = (X, {\mathcal C}_1)$ be a matroid with $X=\{ a, b, c\}$ and ${\mathcal C}_1 = \{ \{a, b\} \}$; let $M_2 = (X, {\mathcal C}_2)$ be the matroid with ${\mathcal C}_2 = \{ \{a, b\}, \{b, c\}, \{a, c\} \}$. In this case, if we add only the circuit $\{b, c\}$ to ${\mathcal C}_1$, the resulting ordered pair is not a matroid, since Item ${ \bf (C3)}$ of Proposition~\ref{propcircuits} fails. The same is true if we add only the circuit $\{a, c\}$. Therefore, it is necessary to add both circuits $\{b, c\}$ and $\{a, c\}$ at the same step in ${\mathcal C}_1$ to ensure that the resulting ordered pair is a matroid, i.e., the matroid $M_2$. Then, for each circuit $\{b, c\}$ and $\{a, c\}$, we assign unique operators $U_{\{b, c\}}$ and $U_{\{a, c\}}$, respectively, at the same step.
\end{example}

\begin{axiom}\label{axiom3c}
    \emph{Universal Circuit Operator}.
    If two matroids $M_1=(X , {\mathcal C}_1 )$ and $M_2=(X , {\mathcal C}_2 )$ differ by a single circuit, i.e., ${\mathcal C}_2 = {\mathcal C}_1 \cup  \{C\}$, then $\ket{M_2} = U_{C} \ket{M_1}$. The circuit operator $U_{C}$ does not depend on either $M_1$ or $M_2$, and it is determined uniquely by the circuit $C$. If $M_2$ is obtained from $M_1$ by adding, at the same step, a minimal finite number of circuits $[C]=\{C_{i_1}, \ldots, C_{i_r}\}$ to satisfy Item~${\bf (C3)}$ of Proposition~\ref{propcircuits}, i.e., ${\mathcal C}_2 = {\mathcal C}_1 \cup  \{[C]\}$, then $\ket{M_2} = U_{[C]} \ket{M_1}$. For each circuit $C_{i_j}$ containing in $[C]$, $j=1, \ldots, r$, there exists an operator $U_{C_{i_j}}$ that does not depend on neither $M_1$ nor $M_2$ and is uniquely determined by $C_{i_j}$, for all $j=1, \ldots, r$. 
\end{axiom}

The procedure for constructing matroid states is similar to that of \cite{Ionicioiu:2012}. Given a matroid $M=(X, {\mathcal C})$, the construction of the matroid state $\ket{M}$ is obtained by starting from the empty matroid and successively applying Axiom~\ref{axiom3c}, thereby resulting in the matroid state.

\begin{eqnarray}\label{quantumstate}
\ket{M} = \displaystyle\prod_{C \in {\mathcal C}}^{} U_{[C]}{\ket{\phi}}^{\otimes n},
\end{eqnarray}
where $[C]$ can be a unique circuit or a finite number of circuits satisfying Item~${\bf (C3)}$ of Proposition~\ref{propcircuits}, according to Axiom~\ref{axiom3c}.

\begin{condition}\label{C1c}\emph{Locality}.
    Let $M=(X, {\mathcal C})$ be a matroid and $C = \{ b_1 , b_2 , \ldots , b_s \} \in {\mathcal C}$. Then the circuit operator $U_{C}$ acts nontrivially only on the Hilbert spaces corresponding to the elements belonging to $C$, i.e., on the space
    ${ \mathcal H}_{b_1} \otimes { \mathcal H}_{b_2} \otimes \ldots \otimes { \mathcal H}_{b_s}$ and acts as identity on the remaining space, i.e., $U_{C}=U_{\{b_1,\ldots,b_s\}}\otimes I^{\otimes (n-s)}$.
\end{condition}

\begin{condition}\label{C2c}\emph{Symmetry}. 
    The circuit operator is symmetric in the inputs, i.e., $U_{C}=U_{C^{*}}$, where $C^{*}$ is a reordering of the elements of $C$.
\end{condition}

\begin{condition}\label{C3c}\emph{Circuit Commutativity}.
    Let $M=(X, {\mathcal C})$ be a matroid and assume that two circuits $C_1$ and $C_2$ have common elements. Then the corresponding operators $U_{C_1}$ and $U_{C_2}$ commute.
\end{condition}

\begin{remark}\label{multioperator}
    In the literature (see, for example, \cite{Ionicioiu:2012,Qu:2013a}), the authors considered a unique operator that acts on two qubits of the graph state \cite{Ionicioiu:2012}, i.e., a triplet $G= ( {\mathcal H}_1, \ket{\psi},U)$. In the case of a hypergraph containing $n+1$ hyperedges, there exist $n+1$ basic operators $U_0, U_1,\ldots, U_n$ corresponding to each of the $n+1$ types of hyperedges \cite{Qu:2013a}. More precisely, if the hyperedge $e_i$ has $i$ elements, then $U_i$ acts on $e_i$. In this paper, we do not assume that the operators corresponding to the circuits are equal. In fact, we only assume that for the construction of a matroid state, we must satisfy Conditions~\ref{C1c}, \ref{C2c}, and \ref{C3c}. We will illustrate this fact in Example~\ref{distinctoperators}.
\end{remark}

The empty matroid is the matroid $M_{n}^{\emptyset}=(X,\emptyset)$, which means that $|X|=n$ and there is no circuit acting on the elements of the ground set. The corresponding quantum state is ${\ket{\psi}}^{\otimes n}$. 

Let $m$ and $n$ be two nonnegative integers with $m \leq n$. Let $X$ be a set with $n$ elements and ${ \mathcal B}$ be the collection of all $m$-element subsets of $X$. For direct computation, it follows that ${ \mathcal B}$ is the set of bases of a matroid on $X$, which is denoted by ${\mathcal U}_{m,n}$, and is called the \emph{uniform matroid} of rank $m$ on an $n$-element set. The set of independent sets ${\mathcal I}({\mathcal U}_{m,n})$ of ${\mathcal U}_{m,n}$ is
\begin{equation}\label{independentuniform}
    {\mathcal I}({\mathcal U}_{m,n}) = \{ A \subseteq X : \ |A| \leq m\}
\end{equation}
and its set of circuits is given by
\begin{equation}\label{circuituniform}
{ \mathcal C}({\mathcal U}_{m,n})= 
\begin{cases}
           \ \  \emptyset, & \ \ if \ \  m=n. \\
           \  \{ A \subseteq X : \ |A| = m +1\}, & \ \ if \ \  m < n.
         \end{cases}
\end{equation}

In what follows, we present an example illustrating Remark~\ref{multioperator}, in which the operators satisfy Conditions~\ref{C1}, \ref{C2}, and \ref{C3} but are different. The well-known Pauli operators are the unitary operators shown below
\begin{align}
\label{Paulis}
    I= {} &
    \begin{pmatrix*}[r]
     1 & 0 \\
     0 & 1\\
\end{pmatrix*}, &
    \sigma_x= {} &
    \begin{pmatrix*}[r]
     0 & 1 \\
     1 & 0\\
\end{pmatrix*}, \nonumber \\
\sigma_y= {} & 
\begin{pmatrix*}[r]
     0 & -i \\
     i & 0\\
\end{pmatrix*}, &
 \sigma_z= {} &  
 \begin{pmatrix*}[r]
     1 & 0 \\
     0 & -1\\
\end{pmatrix*}.
\end{align}

\begin{example}\label{distinctoperators}
Let $M=(E, { \mathcal C})$ be the uniform matroid ${\mathcal U}_{2,4}$, where $E=\{a, b, c, d\}$. We know that the set of circuits of \ ${\mathcal U}_{2,4}$ is\\ 
${ \mathcal C}({\mathcal U}_{2,4}) = \{ \{a, b, c\}, \{a, b, d\}, \{a, c, d\}, \{b, c, d\}\}$. Define: 
\begin{align*}
U_{\{a,b,c\}}= {} & {(\sigma_x\sigma_z)}_{a}\otimes{(\sigma_x\sigma_z)}_{b}\otimes{(\sigma_x\sigma_z)}_{c}\otimes I,\\ 
U_{\{a,b,d\}}= {} & {(\sigma_z\sigma_x)}_{a}\otimes{(\sigma_z\sigma_x)}_{b}\otimes I\otimes {(\sigma_z\sigma_x)}_{d},\\ 
U_{\{a,c,d\}}= {} & {(\sigma_x\sigma_z)}_{a}\otimes I \otimes{(\sigma_x\sigma_z)}_{c}\otimes {(\sigma_x\sigma_z)}_{d},\\ 
U_{\{b,c,d\}}= {} & I \otimes {(\sigma_z\sigma_x)}_{b}\otimes{(\sigma_z\sigma_x)}_{c}\otimes {(\sigma_z\sigma_x)}_{d}.
\end{align*}
Note that these operators satisfy Conditions~\ref{C1c}, \ref{C2c} and \ref{C3c} and $U_{\{a,b,c\}} \neq U_{\{a,b,d\}}$. This implies that, in our context, we can have distinct operators.
\end{example}

Proposition~\ref{completegraph} shown in the following establishes that the intersection between the set of graph states and the set of matroid states is not empty.

\begin{proposition}\label{completegraph}
    Let $n > 1$ be a positive integer. The complete graph state of $n$ qubits is the matroid state of the uniform matroid ${\mathcal U}_{1,n}$.
\end{proposition}
\begin{proof}
    The collection of circuits of ${\mathcal U}_{1,n}$ consists of all sets with exactly two elements. For each circuit $C \in {\mathcal C}({\mathcal U}_{1,n})$, define the circuit operator $U_C = \CZ$, the controlled-$Z$ gate. We then have the graph state $G = ({\mathbb C}^{2},\ket{+},\CZ)$.
\end{proof}

\begin{proposition}\label{proptransitive}
     Let $M=(X, {\mathcal C})$ be a matroid where $|X| \geq 2$. If there exist two operators $U_{\{a, b\}}$ and $U_{\{b, c\}}$ that act nontrivially on the Hilbert spaces corresponding to $\{a, b\}$ and $\{b, c\}$ respectively, then there exists an operator $U_{\{a, c\}}$ that acts nontrivially on $\{a, c\}$.
\end{proposition}

\begin{proof}
    Let $\ket{M}$ be the matroid state derived from $M$ having the operators $U_{\{a, b\}}$ and $U_{\{b, c\}}$. From Axiom~\ref{axiom3c} and Condition~\ref{C1c}, since the operators are only defined in circuits, it follows that $C_1 =\{a, b\}$ and $C_2=\{b, c\}$ are circuits of $M$. By Item $\bf{(C3)}$ of Proposition~\ref{propcircuits}, it follows that there exists a circuit $C_3 \subseteq (C_1 \cup C_2)\setminus \{b\}$. Since circuits are nonempty minimal dependent sets, neither $\{a\}$ nor $\{c\}$ is a circuit. Hence, there exists a unique possibility for $C_3$, i.e., $C_3 =\{a, c\}$. Therefore, again from Axiom~\ref{axiom3c} and Condition~\ref{C1c}, there exists an operator $U_{\{a, c\}}$ that acts nontrivially on $\{a, c\}$.
\end{proof}

The ideas contained in the proof of Proposition~\ref{proptransitive} induce the following result.

\begin{corollary}\label{cortransitive}
     Let $M = (X, { \mathcal C})$ be a matroid, where the collection of circuits ${ \mathcal C}$ consists of all $2$-subsets of $X$ having two elements. Then there exists a transitive relation on $X$.
\end{corollary}
\begin{proof}
    Define the following relation $R$: for $a, b \in X$, $a R b$ if and only if there exists a circuit containing $a$ and $b$. We will prove that $R$ is transitive. Assume that $aRb$ and $bRc$ are true. Since all circuits have two elements, we know that $C_1 =\{a, b\}$ and $C_2 =\{b, c\}$ are circuits. From Item $\bf{(C3)}$ of Proposition~\ref{propcircuits}, there exists a circuit $C_3=\{a, c\} $, e.g., $aRc$.
\end{proof}

\begin{proposition}\label{separablec}
    Let $n$ be an integer, $n \geq 2$, and $M= (X, { \mathcal C})$ be a matroid, with $|X|=n$, ${ \mathcal C}= \{ \{a_1\}, \ldots , \{a_t\}\}$, and $1 \leq t \leq n$. Then the corresponding matroid state $\ket{M}$ is separable.
\end{proposition}
\begin{proof}
    From Condition~\ref{C1c}, for each circuit $\{a_i\} \in { \mathcal C}$ and $i=1, \ldots, t$, there exists a unique operator that acts nontrivially on the Hilbert space corresponding to $a_i$ and as the identity on the remaining Hilbert spaces. Hence, all operators act only locally on a single Hilbert space and, consequently, $\ket{M}$ is separable. 
\end{proof}
It is interesting to note that the reciprocal of Proposition~\ref{separablec} is not true. As a counterexample, consider the uniform matroid ${\mathcal U}_{1,2}$ with $X=\{a,b\}$. The (unique) circuit of \  ${\mathcal U}_{1,2}$ is $\{a, b\}$. From Condition~\ref{C1c}, we assign the operator $U_{a, b}= \sigma_x \otimes \sigma_x$ that acts nontrivially on $\{a, b\}$. Therefore, the matroid state $\ket{M} = \sigma_x\ket{\psi_a}\otimes\sigma_x\ket{\psi_b}$ is separable, but the circuit has more than one element.

In the next result, we show that every graph state can be constructed from a finite set of matroids whose ground sets are the vertices of the corresponding graph.

\begin{proposition}\label{matroidgenerator}
Let $G=(V, E)$ be a graph with $V=\{v_1 , \ldots , v_n\}$. Then there exist matroids $M_{ij}$, all with ground set $V$, whose successive applications of the corresponding circuit operators produce the graph state $\ket{G}$.
\end{proposition}

\begin{proof}
The process is recursive. We start from the empty matroid state $\ket{M_{n}^{\emptyset}}$. Taking an edge $e_{ij}=\{v_{i}, v_{j}\}$ of $G$, we associate it with the ordered pair $M_{ij} =(V, {\mathcal C}_{ij} =\{\{v_i , v_j\}\})$. From Proposition~\ref{propcircuits} and Theorem~\ref{caracircuits}, it follows that $M_{ij}$ is a matroid having ${\mathcal C}_{ij} =\{\{v_i , v_j\}\}$ as the circuit collection. For the unique circuit $\{v_i , v_j\}$ of $M_{ij}$, we assign the operator $U_{\{v_i , v_j\}}=\CZ$ that acts nontrivially as the gate $\CZ$ on the Hilbert spaces associated with the vertices $v_i$ and $v_j$ of $\{v_i , v_j\}$ and the identity on the remaining qubits of $\ket{M_n}$, obtaining the matroid state $\ket{M_{ij}}$. Consider now the matroid state $\ket{M_{ij}}$. For another edge $e_{i' j'}=\{v_{i'}, v_{j'}\}$ of $G$, we again consider the corresponding matroid $M_{i' j' } =(V, {\mathcal C}_{i' j' } =\{\{v_{i' } , v_{j' }\}\})$. For the unique circuit $\{v_{i' }, v_{j'}\}$ of $M_{i' j' }$, we assign the operator $U_{\{v_{i'},v_{j'}\}}=\CZ$ that acts nontrivially as the gate $\CZ$ on the Hilbert spaces associated with the vertices $v_{i' }$ and $v_{j'}$ and the identity on the remaining qubits of the matroid state $\ket{M_{ij}}$, obtaining a new $n$ qubit state $\ket{Q_2}$ (which is not necessarily a matroid state). Similarly, starting from $\ket{Q_2}$, we take another edge $e_{i'' j'' }=\{v_{i'' }, v_{j'' }\}$ of $G$ and consider the matroid $M_{i'' j'' } =(V, {\mathcal C}_{i'' j''} =\{\{v_{i'' } , v_{j'' }\}\})$ associated with $e_{i'' j'' }$. Again, we assign to the circuit $\{v_{i'' }, v_{j''}\}$ the operator $U_{\{v_{i''},v_{j''}\}}=\CZ$ that acts nontrivially as the gate $\CZ$ on the Hilbert spaces associated with the vertices $v_{i'' }$ and $v_{j'' }$ and the identity on the remaining qubits of the quantum state $\ket{Q_2}$, obtaining another quantum $n$ qubit state $\ket{Q_3}$ (which is not necessarily a matroid state).
We repeat the process for any edge of $G$; since the number of steps in this process is finite, we obtain an $n$ qubit state which is exactly the graph state $\ket{G}$. 
Therefore, 
$$\ket{G}= \prod_{\{v_i , v_j \} \in E} U_{C_{ij}} \ket{M_{n}^{\emptyset}},$$ 
where $C_{ij}= \{v_i , v_j \} $. The proof is complete.
\end{proof}

\begin{remark}\label{remarkgraph}
    One can ask the question: What is the difficulty in finding a matroid state corresponding to each graph state in the case where the ground set of the matroid is the set of vertices of the respective graph? To answer this question, we give a simple example: consider a graph $G=(V, E)$, where $|V|= \{v_1,v_2,v_3\}$ and $E=\{ \{v_1,v_2\}, \{v_2,v_3\} \}$. The corresponding graph state has two gates $\CZ$ acting nontrivially on both $\{v_1 , v_2\}$ and $\{v_2 , v_3\}$ and the identity operator on the remaining qubit. Assume there exists a matroid $M_G= (V, {\mathcal C})$, where ${\mathcal C}$ is the collection of circuits of $M_G$, such that $\ket{M_G}=\ket{G}$. From Item ${ \bf (C3)}$ of Proposition~\ref{propcircuits}, the set $\{v_1,v_3\}$ must also be a circuit of $M_G$, since the sets $\{v_1\}$ and $\{v_3\}$ are both independent due to the definition of a circuit. Therefore, in $\ket{M_G}$, there must be three gates $\CZ$ acting nontrivially on the corresponding Hilbert spaces associated with the three circuits $\{v_1, v_2\}$, $\{v_2,v_3\}$, and  $\{v_1,v_3\}$, and the identity operator on the remaining qubit, which implies that $\ket{M_G} \neq \ket{G}$, a contradiction. This counterexample shows that the structure of a matroid does not allow us to ``choose'' every collection of $2$-sets as circuits.
\end{remark}

The following result establishes a connection between graphs and matroids.

\begin{proposition}\label{characterization}
Let $G=(V, E)$ be a graph. The ordered pair $M_G=(V, {\mathcal C})$, where ${\mathcal C}$ is the collection of $2$-subsets $C_{u v}=\{u , v\}$ of $V$ that correspond to all edges $\{u , v \}$ of $G$ is a matroid with the collection of circuits ${ \mathcal C}$ if and only if for all edges of $G$ such that $\{u, v\}, \{v, w\} \in E$, it implies $\{u, w\} \in E$.
\end{proposition}

\begin{proof}
Assume that $M_G=(V, { \mathcal C})$ is a matroid with the collection of circuits ${\mathcal C}$ consisting of the $2$-subsets of $V$ corresponding to all edges of $G$. Suppose $\{u, v\}$ and $\{v, w\}$ are edges of $G$. Then there exist circuits $C_{uv}=\{u, v\}$ and $C_{vw}=\{v, w\}$ of $M_G$. From Item {\bf (C3)} of Proposition~\ref{propcircuits}, it follows that there exists a circuit $C$ such that $C \subseteq [(C_{uv} \cup C_{vw})\setminus\{v\}] = \{u, w\}$. The unique possibility for the circuit $C$ is $C= \{u, w\}$. Again, from the definition of $M_G$, it follows that $\{u, w\}$ is an edge of $G$, i.e., $\{u, w\} \in E$.
    
On the other hand, let $M_G = (V, {\mathcal C})$ be the ordered pair, where the elements of ${\mathcal C}$ are the $2$-subsets $C_{u v}=\{u , v\}$ of $V$ that correspond to all edges $\{u , v \}$ of $G$. In order to show that $M_G$ is a matroid, we prove that ${ \mathcal C}$ satisfies Items {\bf (C1)}, {\bf (C2)}, and {\bf (C3)} of Proposition~\ref{propcircuits}, after applying Theorem~\ref{caracircuits}. Items {\bf (C1)} and {\bf (C2)} clearly hold. We next prove {\bf (C3)}. If $C$ and $C^{*}$ are distinct members of ${\mathcal C}$ and $C \cap C^{*} =\emptyset$, there is nothing to prove. Let $C$ and $C^{*}$ be distinct members of ${\mathcal C}$ such that $v \in C \cap C^{*}$. From the definition of ${ \mathcal C}$, $C=C_{uv}=\{u, v\}$ and $C^{*}=C_{vw}=\{v, w\}$, where $\{u, v\}$ and $\{v, w\}$ are edges of $G$. From the hypothesis, it follows that $\{u, w\} \in E$. Again, from the definition of $M_G$, we have $C_{uw}=\{u, w\} \in { \mathcal C}$ and  $C_{uw}\subseteq (C \cup C^{*})\setminus \{v\}$. Hence, $M_G$ is a matroid. The proof is complete.
\end{proof}

Corollary~\ref{graphmatro}, shown in the sequence, means that there exist certain types of graph $G$ that induce matroids $M_G$ such that the corresponding graph state is equal to the corresponding matroid state.

\begin{corollary}\label{graphmatro}
     Let $G=(V, E)$ be a graph that satisfies the following property: for all edges of $G$ such that $\{u, v\}, \{v, w\} \in E$, it implies $\{u, w\} \in E$. Then there exists a matroid $M_G$ such that $\ket{G} = \ket{M_G}$.
\end{corollary}

\begin{proof}
    From Proposition~\ref{characterization}, it follows that $M_G$ is a matroid whose circuits are the $2$-subsets of the edges of $G$. For each circuit $C_{uv}=\{u, v\}$ of $M_G$, we assign the gate $\CZ$ that acts nontrivially as the gate $\CZ$ on the qubits corresponding to $u$ and $v$, and as the identity on the remaining qubits. From the construction, one has $\ket{G}=\ket{M_G}$. In particular, if the edges of $G$ are mutually disjoint, e.g, they do not share a common vertex, then $\ket{G} = \ket{M_G}$.
\end{proof}

It is known that any matroid is a hypergraph: it is sufficient to assign the ground set $X$ of the matroid $M= (X, { \mathcal C})$ to the set of vertices $V$ of the hypergraph $H=(V, {\mathcal E})$, and each circuit $C \in { \mathcal C}$ of $M$ to a hyperedge $E \in {\mathcal E}$. Therefore, any matroid state is, in fact, a hypergraph state if we consider that the operators given for the corresponding circuits (except for operators that act nontrivially on a single qubit) are the same generalized $\CCZ{|E|-1}\equiv \mathbbm{1}-2\ketbra{1 \ldots 1}$ gate that acts ${\sigma}_{z}$ on one of the qubits conditioned on $|E|-1$ others being $1$ \cite{Liu:2022}.
Then, what are the advantages of applying matroid states rather than hypergraph states? The answer is that the structure of a matroid $M$ provides much more information about the corresponding matroid state compared to hypergraph states (see, for instance, Propositions~\ref{proptransitive} and \ref{operatorproperty}, \ref{separablec}, \ref{matroidgenerator}, \ref{characterization}).
A natural question is whether there exists a matroid state that cannot be reproduced from graph states using only local operations. The answer to this question is yes. In fact, if we consider the matroid $M=(X, {\mathcal C})$, where $X=\{a, b, c\}$ and ${\mathcal C}= \{ \{a\}, \{b\} \}$, then the corresponding matroid state is separable, while all graph states, except for the empty graph state, are entangled. Note that it is reasonable not to consider the application of local operators because the operators that act on graph states are $2$-qubit operators.

\section{Matroid States from Independent Sets}
\label{sec:mssets}

The first task is to define the concept of separability of matroid states induced by independent sets. At first glance, it seems that the more natural way to do this is to use the matroid direct sum. Let us then recall the direct sum of matroids in terms of independent sets.
Let $M_1=(X_1 , {\mathcal I}_1 )$ and $M_2=(X_2 , {\mathcal I}_2 )$ be matroids on disjoint sets $X_1$ and $X_2$, and let $X= X_1 \cup X_2$. 
Then the matroid $(X, {\mathcal I})$, where ${\mathcal I}= \{I_1 \cup I_2 | I_1 \in {\mathcal I}_1 , \ I_2 \in {\mathcal I}_2 \}$ is called the direct sum of $M_1$ and $M_2$, is denoted by $M_1 \oplus M_2$.
Note that the characteristic of the independent sets of the direct sum $M_1 \oplus M_2$ is not convenient for defining the separability of matroid states since independent sets of $M_1 \oplus M_2$ are “mixed", and this fact allows us to define entanglement operators from these independent sets.

As a simple example to see this, consider two disjoint matroids $M_1 = (X_1, { \mathcal I}_1)$, where $X_1 =\{a\}$ and ${ \mathcal I}_1 = \{ \emptyset, \{a\}\}$, and $M_2 = (X_2,{\mathcal I}_2)$, where $ X_2 = \{b\}$ and $ { \mathcal I}_2 = \{ \emptyset, \{b\}\}$. It then follows that the collection of independent sets of the direct sum is given as ${\mathcal I}_{M_1 \oplus M_2}= \{ \emptyset, \{a\}, \{b\}, \{a, b\}\}$. Putting the operator $U_{a, b} = \CZ$, we will have an entangled matroid state $\ket{M}$, while  $M_1 \oplus M_2$ is always disconnected. Consequently, in this context, where operators must be assigned to independent sets, separability cannot be described by $M_1 \oplus M_2$. In other words, when considering independent sets, we need a different axiomatization of separability for matroid states.

To motivate this new manner of axiomatizing separability, recall that the sum $G \Delta G'$ of two graphs (hypergraphs) $G = (V, E)$ and $G' = (V',E')$ is the graph (hypergraph) defined as $G \Delta G' = (V \cup V',E \Delta E')$, where $\Delta$ is the symmetric difference of $E$ and $E'$, i.e., $E \Delta E' = E \cup E' \setminus E \cap E'$. In the case of a direct sum, e.g., $V$ and $V'$ being disjoint, we have $E \Delta E' = E \cup E'$. Then there exist $|E| + |E'|$ edges of $G \Delta G'$, where the edges of $E$ are disjoint from the edges of $E'$. Therefore, there are no mixed edges.
Similarly, the set of circuits of $M_1 \oplus M_2$ is the union $ { \mathcal C}_1 \uplus { \mathcal C}_2 $ of $ {\mathcal C}_1$ and  ${\mathcal C}_2$. This implies that there exist $| { \mathcal C}_1 | + |{ \mathcal C}_2|$ circuits, where the circuits of ${ \mathcal C}_1$ are disjoint from the circuits of ${ \mathcal C}_2$. Again, there are no mixed circuits. Based on these facts, given two matroids $M_1=(X_1 , {\mathcal I}_1 )$ and $M_2=(X_2 , 
{\mathcal I}_2 )$ on disjoint sets $X_1$ and $X_2$, we define the disjoint union $M_1 \uplus M_2$ of $M_1$ and $M_2$ as $$M_1 \uplus M_2 = ( X_1 \cup X_2 , {\mathcal I}_1 \cup{\mathcal I}_2 ).$$
We know that $M_1 \uplus M_2$ is not a matroid, but the cardinality of the set ${\mathcal I}_1 \cup{\mathcal I}_2$ is the sum of the cardinalities of ${\mathcal I}_1$ and ${\mathcal I}_2$, and there are no mixed independent sets. Therefore, the most natural way to define the separability of matroid states is by means of $M_1 \uplus M_2$, because although it is not a
matroid itself, it is defined in terms of two matroids. 

\begin{axiom}\label{axiom1s}
	\emph{Separability}. Let $M_1=(X_1 , {\mathcal I}_1 )$ and $M_2=(X_2 , 
	{\mathcal I}_2 )$ be two matroids on disjoint sets $X_1$ and $X_2$. 
	Then one has $\ket{M_1 \uplus M_2} = \ket{M_1}\otimes\ket{M_2}$.
\end{axiom}

The following axioms and results are similar to those presented in Section~\ref{sec:mscirc}. 

\begin{definition}\label{emptyindep}
    Let $M_{n}^{\emptyset} =(X, {\mathcal I})$ be a matroid with ${\mathcal I}=\{ \emptyset \}$ and $|X|=n$. Then $M_{n}^{\emptyset}$ is called the empty matroid on $n$ elements.
\end{definition}

\begin{corollary}\label{coro1}
    If $M_{n}^{\emptyset}=(X, \{ \emptyset \} )$ is the empty matroid, where $|X|=n$, then $\ket{M_{n}^{\emptyset}} = \ket{{\phi}_1}\otimes \ket{{\phi}_2}\otimes \ldots \otimes \ket{{\phi}}_n$.
\end{corollary}

\begin{corollary}\label{coro2}
    Given a matroid $M=(X, { \mathcal I})$ with $X=\{a_1 , \ldots , a_n\}$, we associate with each element $a_i \in X$, $i=1, \ldots,n$ a Hilbert space ${\mathcal H}_i$. The total Hilbert space is given by ${\mathcal H} = \displaystyle \bigotimes_{i=1}^{n}{\mathcal H}_i$.
\end{corollary}
\begin{proof}
    Since by Corollary~\ref{coro1} the empty matroid is mapped to a product state, this induces the idea of mapping elements of the ground set $X$ to Hilbert spaces, because each $i \in X$ has an associated quantum state $\ket{\phi}_i$.
\end{proof}

\begin{axiom}\label{axiom2}
    \emph{Matroid Isomorphism}. Let $M_1=(X_1 , {\mathcal I}_1 )$ and $M_2=(X_2 , {\mathcal I}_2 )$ be two isomorphic matroids. Then the corresponding density operators ${\rho}_1 = \ketbra{M_1}$ and ${\rho}_2 = \ketbra{M_2}$ satisfy
\begin{eqnarray}
    {\rho}_2 = D( P_{1,2}) {\rho}_1 [D( P_{1,2})]^{-1},
\end{eqnarray}
     where $D( P_{1,2})$ is a matrix representation of the bijection $P_{1,2}$ mapping $M_1$ to $M_2$.
\end{axiom}

\begin{corollary}\label{coro3}
    If $P_M$ is an automorphism of the matroid $M=(X, { \mathcal I})$ then 
    $\left[ \rho , D(P_{M} )\right] =0$.
\end{corollary}

\begin{proposition}\label{prop1}
    Let $M=(X, {\mathcal I})$ be a matroid, where $|X|=n$. Then the corresponding matroid state $\ket{M}$ belongs to a Hilbert space ${ \mathcal H}$ of $n$ identical quantum systems ${\mathcal H}_1$, i.e., ${\mathcal H}=  {\mathcal H}_1^{\otimes n}$, where ${\mathcal H}_1$ is a Hilbert space associated with a single element of the ground set $X$. Additionally, the empty matroid $M_{n}^{\emptyset}$ is assigned to $M_{n}^{\emptyset} \rightarrow \ket{M_{n}^{\emptyset}} ={\ket{\phi}}^{\otimes n}$, where ${\phi} \in {\mathcal H}$.
\end{proposition}
   
We next state an analogue of Axiom~\ref{axiom3c} for matroids derived from independent sets.

\begin{axiom}\label{axiom3}
    \emph{Universal independent set operator}.
    If the matroids $M_1=(X , {\mathcal I}_1 )$ and $M_2=(X , {\mathcal I}_2 )$ differ by a single independent set, i.e., ${\mathcal I}_2 = {\mathcal I}_1 \cup  \{I\}$, then $\ket{M_2} = U_{I} \ket{M_1}$. The independent set operator $U_{I}$ does not depend on either $M_1$ or $M_2$; it is uniquely determined by the independent set $I$. If $M_2$ is obtained from $M_1$ by adding, at the same step, a minimal finite number of independent sets $[I]=\{I_{i_1}, \ldots, I_{i_t}\}$ satisfying Definition~\ref{matro}, e.g., ${\mathcal I}_2 = {\mathcal I}_1 \cup  \{[I]\}$, then $\ket{M_2} = U_{[I]} \ket{M_1}$. For each independent set $I_{i_j}$ contained in $[I]$, $j=1, \ldots, t$, there exists an operator $U_{I_{i_j}}$ that does not depend on neither $M_1$ nor $M_2$ and is uniquely determined by $I_{i_j}$, for all $j=1, \ldots, t$. 
\end{axiom}

To guarantee the consistency of Eq.~(\ref{quantumstate}), the independent set operators must satisfy the following three conditions. The first one asserts that the operators act nontrivially on the Hilbert spaces associated with elements of independent sets.

\begin{condition}\label{C1}\emph{Locality}.
    Let $M=(X, {\mathcal I})$ be a matroid and $I^{*} = \{ a_1 , a_2 , \ldots , a_t \} \in {\mathcal I}$. Then the independent set operator $U_{I^{*}}$ acts nontrivially only on the Hilbert spaces corresponding to the elements belonging to $I^{*}$, i.e., on the space ${ \mathcal H}_{a_1} \otimes { \mathcal H}_{a_2} \otimes \ldots \otimes { \mathcal H}_{a_t}$, and acts as identity on the remaining space, that is, $U_{I^{*}}=U_{\{a_1, \ldots,a_t\}}\otimes I^{\otimes (n-t)}$.
\end{condition}

\begin{remark}
    It is worth mentioning that, by Proposition~\ref{prop1}, the quantum systems ${\mathcal H}_{a_i}$ with $i=1, 2, \ldots, n$ are equal. In Condition~\ref{C1} given previously, we utilize the notation ${ \mathcal H}_{a_i}$ to denote that an operator acts on the quantum system corresponding to the element $a_i$.
\end{remark}

\begin{condition}\label{C2}\emph{Symmetry}. 
    Let $M= (X, {\mathcal I})$ be a matroid and $I \in { \mathcal I}$. The independent set operator is symmetric in the inputs, i.e., $U_{I}=U_{I^{*}}$, where $I^{*}$ is a reordering of the elements of $I$.
\end{condition}

The commutativity condition is fundamental to allow the application of operators in any order.

\begin{condition}\label{C3}\emph{Independent Set Commutativity}.
    Let $M=(X, {\mathcal I})$ be a matroid and assume that two independent sets $I_1$ and $I_2$ have common elements. Then the corresponding operators $U_{I_1}$ and $U_{I_2}$ commute.
\end{condition}

An important advantage of considering matroid states instead of graph or hypergraph states is that all the Hilbert spaces where the operators will act nontrivially are completely determined by knowing the bases of the underlying matroid. For example, if the ground set of a matroid (or the set of edges (hyperedges) of a graph (hypergraph)) is large, then we must assign an operator to each edge (hyperedge). In the case of a matroid $M$, even with a large ground set, it is only necessary to have the bases of $M$ to know all sets of Hilbert spaces in which the operators act nontrivially. This fact provides a much more efficient process for assigning the corresponding operators.

\begin{proposition}\label{prop2}
    Let $M=(X, {\mathcal I})$ be a matroid. Then, the subsets of $X$ under which the operators will be applied are completely determined by the collection of bases of $M$.
\end{proposition}

\begin{proof}
    Assume that ${\mathcal B}$ is the set of bases of $M$. From Axiom~\ref{axiom3}, for any $B \in {\mathcal B}$, there exists a unique independent set operator $U_{B}$ determined by $B$.
    From Item~${\bf (I2)}$ of Definition~\ref{matro}, all subsets of $B$ are also independent, which implies that they have a corresponding independent set operator. In other words, since the bases define all independent sets, it follows that the subsets of $X$ under which all the operators are applied are completely determined from the collection of bases of $M$.
\end{proof}

\begin{proposition}\label{operatorproperty}
    Let $M=(X, {\mathcal I})$ be a matroid of rank $k$. Then, for each pair of operators $U_{I_i}$ and $U_{I_j}$, where $1 \leq i, j \leq k$, with $|I_j| >1$, there exists at least one operator that acts nontrivially on a common Hilbert space on which $U_{I_i}$ or $U_{I_j}$ acts nontrivially.
\end{proposition}
\begin{proof}
    Suppose that $I_i$ and $I_j$ are independent sets in which the operators $U_{I_i}$ and $U_{I_j}$ are assigned, respectively.\\ 
    Case 1: $|I_j| \neq |I_i|$ . Assume without loss of generality that $|I_j| > |I_i|$ (if $|I_i| > |I_j|$, then $|I_i| >1$, and the procedure adopted is the same). Then it follows from Item~${ \bf (I3)}$ of Definition~\ref{matro} that there exists an element $a \in I_j \setminus I_i $ such that $I_i \cup \{a\} \in { \mathcal I}$. From Axiom~\ref{axiom3} and Condition~\ref{C1}, there exists an operator $U_{I_i \cup \{a\}}$ that acts nontrivially on the Hilbert space associated with $I_i \cup \{a\}$. Therefore, $U_{I_j}$ and $U_{I_i \cup \{a\}}$ both act nontrivially on the Hilbert space corresponding to $a$.\\
    Case 2: $|I_j| = |I_i|$. Assume that $|I_j| = |I_i| >1$. Let $b$ be an element of $I_i$; from Item~${ \bf (I2)}$ of Definition~\ref{matro}, $\{b\} \in {\mathcal I}$. Applying Case 1 for the independent sets $\{b\}$ and $I_j$, there exists an element $c \in I_j \setminus \{b\}$ such that $\{b, c\} \in { \mathcal I} $. The operators $U_{I_j}$ and $U_{\{b, c\}}$ act nontrivially on the Hilbert space corresponding to $c$, and the proof is complete.
\end{proof}

\begin{corollary}\label{disjointop}
    Let $M=(X, {\mathcal I})$ be a matroid such that ${\mathcal I}$ has at least one independent set with more than one element. Then the set of operators is not disjoint, i.e., there exists at least a pair of operators acting nontrivially on the same set of Hilbert spaces.
\end{corollary}
\begin{proof}
    The proof is immediate from Proposition~\ref{operatorproperty}.
\end{proof}

\begin{remark}
    Note that Proposition~\ref{operatorproperty} and Corollary~\ref{disjointop} suggest that in the case where independent sets are considered to construct matroid states, there is a large possibility of generating entangled matroid states, since independent set operators ``mix" these quantum states. In other words, Items ($\bf{I2}$) \ and \ ($\bf{I3}$) of Definition~\ref{matro} are suitable for creating entangled matroid states if we consider the $\CZ$ gate or the generalized $\CCZ{|E|-1}$ gate.
\end{remark}

\section{Matroid States and Stabilizer Formalism}
\label{sec:msstab}

In this section, we review some concepts of stabilizer groups relevant to the construction of graph states. For more details, we refer the reader to \cite{Hein:2006}; see also \cite{Gottesman:1996,Gottesman:1998,Calderbank:1998,Nielsen:2010,Ketkar:2006} for a complete presentation of the stabilizer formalism.

Recall that the Pauli group on $1$-qubit is the set $G_1 =\{ \pm I, \pm iI, \pm \sigma_x, \pm i\sigma_x, \pm \sigma_y, \pm i\sigma_y, \pm \sigma_z, \pm i\sigma_z\}$ under the operation of matrix multiplication, where $I$, $\sigma_x$, $\sigma_y$, and $\sigma_z$ are the Pauli operators shown in Eq.~(\ref{Paulis}). In the same line, the general Pauli group $G_n$ on $n$ qubits consists of all $n$-fold tensor products of Pauli matrices, together with the multiplicative factors $\pm 1$ and $\pm i$.

Suppose that $S$ is a subgroup of $G_n$; define $V_S$ as the set of $n$ qubit states that are fixed by every element of $S$. Then $V_S$ is a vector space, called the \emph{vector space stabilized by} $S$. The subgroup $S$ is said to be the \emph{stabilizer} of the vector space $V_S$.

The following result is a powerful tool for constructing (quantum) stabilizer codes.

\begin{proposition}\cite[Proposition 10.5]{Nielsen:2010}
    Let $S= \langle g_1, \ldots, g_{n-k}\rangle \subset G_n$ be the subgroup generated by $n-k$ independent and commuting elements from $G_n$ such that $-I \notin S$. Then $V_S$ is a $2^k$-dimensional vector space.
\end{proposition}
 
In particular, if we consider only one quantum state $\ket{\psi}$, we can ask if there exists a set of stabilizers that stabilizes $\ket{\psi}$. In this context, it was shown that the graph states are properly stabilized by the tensor product of Pauli operators, as the next result shows.

\begin{proposition}\cite[Proposition 2]{Hein:2006}\label{Heinstabilizer}
    Let $G=(V, E)$ be a graph. A graph state vector $\ket{G}$ is the unique, common eigenvector in 
     $({ \mathbb C}^2)^V $ to the set of independent commuting observables: 
     $$K_a= (\sigma_x)_{a} (\sigma_z)_{N_a}:= (\sigma_x)_{a} \prod_{b \in N_a}(\sigma_z)_{b},$$
     where the eigenvalues for the correlation operators $K_a$ are equal to $+1$ for all $a \in V$ and $N_a$ denotes the neighborhood of vertex $a$, i.e., $N_a =\{b \in V | \ \{a, b\} \in E\}$. The abelian subgroup $S$ of the local Pauli group $G_{|V|}$ generated by the set $\{K_a | \ a \in V\}$ is said to be the stabilizer of the graph state.
\end{proposition}

In the following result, we present a method for obtaining a graph state from its stabilizer group, which is constructed from a finite set of matroid states.

\begin{proposition}\label{correspgraphmatro}
     Let $G=(V, E)$ be a graph with $n$ vertices. Then there exist \ $t+1$ matroid states whose corresponding operators are the generators of the stabilizer subgroup of the graph state $\ket{G}$, where $t$ is the number of incident vertices of $G$.
\end{proposition}

\begin{proof}
    Let $G=(V, E)$ be the underlying graph with $V=\{a_1 , \ldots , a_n\}$, which corresponds to the graph state $\ket{G}$ on $n$-qubits. Let ${\operatorname{Inc}}_G=\{a_{i_1}, a_{i_2}, \ldots, a_{i_t}\}$ be the set of incident vertices of $G$, i.e., the set of vertices that are incident to some edge of $G$. For each $a_{i_j} \in {\operatorname{Inc}}_G$, consider the corresponding neighborhood $N_{a_{i_j}}$, $j = 1, \ldots , t$, and for each neighborhood $N_{a_{i_j}}$ of $a_{i_j}$, consider the corresponding matroid $M_{a_{i_j}}= (V, {\mathcal C}=\{N_{a_{i_j}}\})$, with a unique circuit $N_{a_{i_j}}$.
     
     By Axiom~\ref{axiom3c} and Condition~\ref{C1c}, we assign to each circuit $N_{a_{i_j}}$ a unique operator (in the corresponding matroid state) that acts nontrivially on the qubits corresponding to $N_{a_{i_j}}$ of the form 
     \begin{align*}
     U_{N_{a_{i_j}}}= {}  &\left[\displaystyle\prod_{ b \in N_{a_{i_j}}}^{} (\sigma_z)_{b}\right] \nonumber 
     \otimes I^{\otimes (n - |N_{a_{i_j}}|)},
     \end{align*}
     for all $j = 1, \ldots , t$. 
     
     We next define the matroid $M_{V_G}=(V, { \mathcal C}_{M_{V_G}})$, where $V=\{a_1 , \ldots , a_n\}$ and $ { \mathcal C}_{M_{V_G}}= \{ \{a_1\}, \{a_{2}\}, \ldots , \{a_n\}\}$. From Axiom~\ref{axiom3c} and Condition~\ref{C1c}, we assign to each circuit $\{a_i\}$, $i=1 , \ldots , n$, the operator $U_{a_{i}}$, in the corresponding matroid state, that acts as Pauli $\sigma_x$ in the qubit corresponding to $a_{i}$ and as identity on the remaining qubits, i.e., $U_{a_{i}} = (\sigma_x)_{a_i}\otimes I^{\otimes (n-1)}$, for all $i=1 , \ldots , n$. We know that for all vertices $a_{i_1}, a_{i_2}, \ldots , a_{i_t}$ in ${\operatorname{Inc}}_G$ we have 
     $$K_{a_{i_j}}= U_{a_{i_j}}\circ U_{N_{a_{i_j}}}.$$
     For vertices $a_{i_t + 1}, \ldots, a_n$ that are not incident with any edge, the unique action that such vertices undergo is from the matroid state operators that act on these corresponding qubits of the form 
     $$K_{a_{r}} = U_{a_{r}} = (\sigma_x)_{a_r}\otimes I^{\otimes (n-1)},$$ $r= i_t+1 , \ldots , n$. 
     By Proposition~\ref{Heinstabilizer}, the graph state $\ket{G}$ is the unique state stabilized by the subgroup $S$ generated by the operators  $\{ K_a \ | \ a \in V\}$. 

    Therefore, for each graph state $\ket{G}$ on $n$ qubits, there exist $t+1$ matroids $M_{a_{i_j}}$, $j= 1, \ldots, t$, and $M_{V_G}$, whose operators or the composite of operators (the latter is the case of the incident vertices $a_{i_1}, a_{i_2}, \ldots, a_{i_t}$, which correspond to the composite operators $K_{a_{i_j}}= U_{a_{i_j}}\circ U_{N_{a_{i_j}}}$, $j = 1, \ldots, t$) are independent commuting generators of the subgroup $S$ of $G_n$ that stabilize $\ket{G}$. The proof is complete.
\end{proof}

\section{Hypergraph, matroid, graph, and stabilizer states}
\label{sec:relations}

\begin{figure}
    \centering
    \includegraphics[width=\columnwidth]{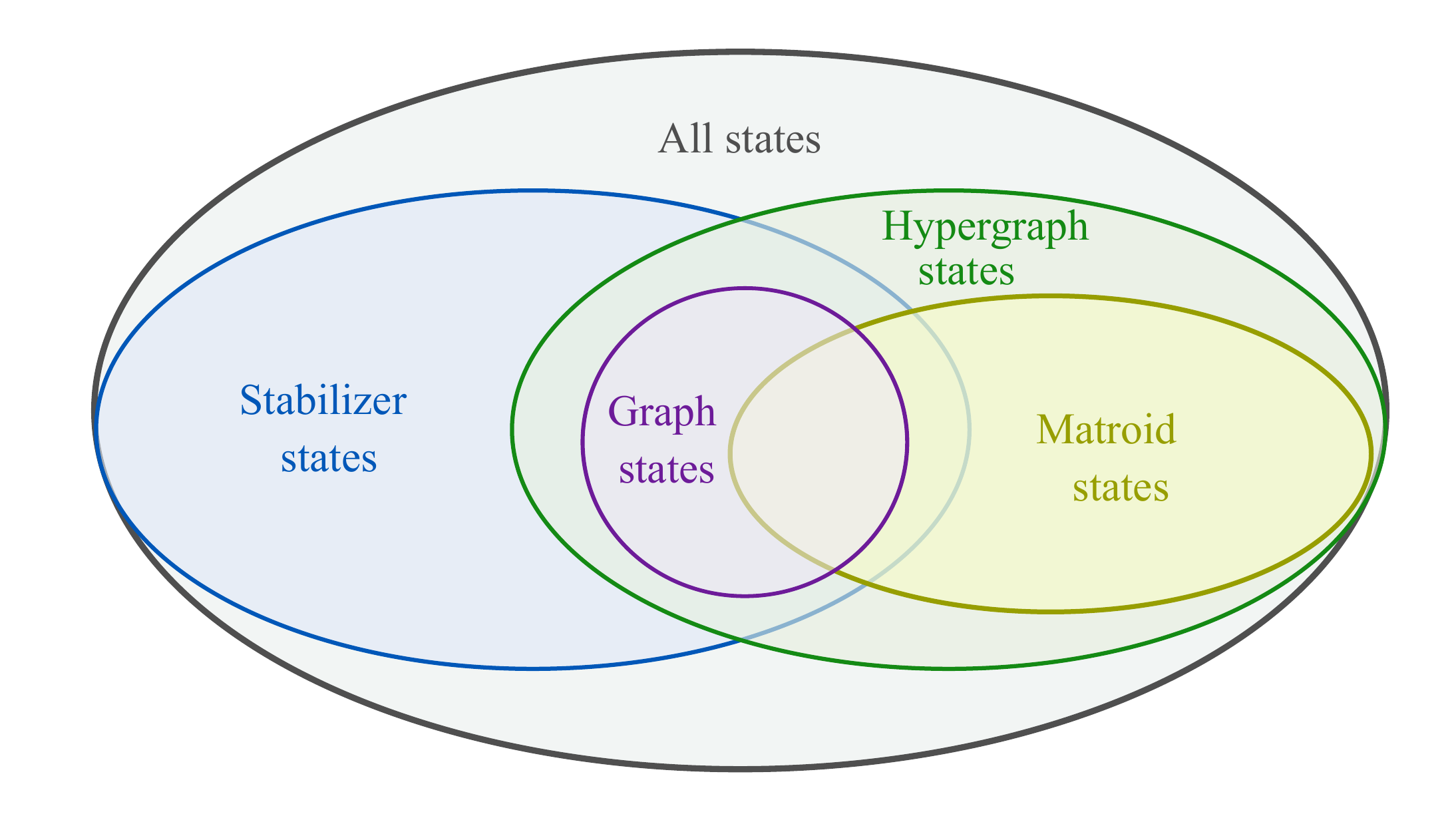}
    \caption{The relations among hypergraph states, matroid states, stabilizer states, and graph states.}
    \label{fig:classes1}
\end{figure}

From the previous sections, it is possible to conclude that there are relations between the classes of states defined by different frameworks, namely graph states, hypergraph states, stabilizer states, and matroid states. These relations are shown in Fig. \ref{fig:classes1}. Here, we are not considering local unitary transformations between the states.

In summary, every matroid state is a hypergraph state because it is always possible to associate an independent set or a circuit with a hyperedge, thereby recovering the hypergraph structure. However, the inclusion of matroid states in hypergraph states is strict because there are hypergraph states that cannot be represented as matroid states. On the other hand, this hierarchy does not work on the relation between graph states and matroid states, as discussed in \text{Remark} \ref{remarkgraph}, except for the specific case presented in \text{Corollary} \ref{graphmatro}. Finally, in \cite{Qu:2013a}, it was shown that graph states are a class of stabilizer states. Therefore, there is an intersection between the sets of matroid states and stabilizer states. Furthermore, given an arbitrary graph, we have shown a method to construct a finite collection of matroid states whose corresponding operators are the generators of the stabilizer subgroup of the underlying graph state in Proposition \ref{correspgraphmatro}.

\section{Discussion}
\label{sec:disc}

As is well known, interest in graph and hypergraph states began with a particular class of entanglement states: cluster states. Because of this, among other reasons, when the theory of graph states was formalized \cite{Ionicioiu:2012}, the authors attributed the $\CZ$ gate to each edge of the underlying graph, since such an operator produces the maximum entanglement state. After this, such a theory was generalized by hypergraphs in \cite{Qu:2013a}. Building on these starting points, extensive research has explored the underlying structures of graphs and hypergraphs to obtain maximum entanglement states. 

In this work, we have proposed a new approach to quantum states derived from matroids: matroid states, which are quantum states associated with a given matroid. An advantage of this construction is that the underlying matroid structure yields desirable properties of the corresponding matroid state. Unlike in the cases of graph and hypergraph states, here we do not assume that the operators must be the same, or even the $\CZ$ or $\CCZ{|E|-1}$ gate in particular. Of course, if we are interested in entangled states, we can attribute to each independent circuit or set an operator that provides this entanglement, i.e., the $\CZ$ or $\CCZ{|E|-1}$ gates, except in the case where the operator acts nontrivially only on one qubit. On the other hand, if we are interested in other quantum properties of the quantum state, the freedom of choice for operators associated with independent sets is essential. For example, we can consider magic as a quantum resource. If we are restricted to the $\CZ$ gate, it is only possible to generate states in the stabilizer framework and therefore without magic. In our case, we work in a more general scenario where the operators must be commutative among themselves, local, and symmetric in their inputs. Still, we do not assume that the operators are equal. Because a matroid has a structure that generalizes vector spaces, graphs, algebraic dependence, and so on, we utilize this structure to encode interesting properties about the corresponding matroid state.

\section{Final Remarks}
\label{sec:fr}

We have introduced matroid states, quantum states derived from matroids, whose quantum operators are assigned to circuits or independent sets. The inherent structure of a matroid offers some advantages over that of graphs and hypergraphs. For example, the set under which the operators are applied is completely determined by the basis; the matroid structure provides suitable properties concerning the corresponding matroid state; and we can obtain graph and hypergraph states from a set of matroid states. However, since we do not require the operators to be the same, our theory encompasses the theory of graph or hypergraph states. Although we consider that the operators can be distinct, since the final quantum state must be unique for each circuit or independent set, we fix a unique operator that acts nontrivially on the corresponding individual Hilbert spaces, such that it commutes with all others and is symmetric in its inputs.

\emph{Note:} During the completion stages of this work, we became aware of a related work  \cite{Huang:2026} proposing quantum states supported by matroids.

\section{Acknowledgements}
This work was partially supported by Coordenação de Aperfeiçoamento de Pessoal de Nível Superior (CAPES, Finance Code 001).
It was also supported by Conselho Nacional de Desenvolvimento Científico
Tecnológico (CNPq).
F.M.A. acknowledges financial support from Fundação Araucária Project No. 305 and CNPq Grant No. 313124/2023-0, and G.G.LaG. acknowledges financial support from CNPq Grant No. 302984/2022-4.\\

\noindent\textbf{Conflict of interest}\\
The authors declare no conflicts of interest.\\

\noindent\textbf{Data availability}\\
No new data were created or analyzed in this study. This is purely a theoretical work, and all mathematical derivations and physical conclusions are explicitly detailed within the manuscript.

\bibliography{bib.bib}

\begin{thebibliography}{24}%
\makeatletter
\providecommand \@ifxundefined [1]{%
 \@ifx{#1\undefined}
}%
\providecommand \@ifnum [1]{%
 \ifnum #1\expandafter \@firstoftwo
 \else \expandafter \@secondoftwo
 \fi
}%
\providecommand \@ifx [1]{%
 \ifx #1\expandafter \@firstoftwo
 \else \expandafter \@secondoftwo
 \fi
}%
\providecommand \natexlab [1]{#1}%
\providecommand \enquote  [1]{``#1''}%
\providecommand \bibnamefont  [1]{#1}%
\providecommand \bibfnamefont [1]{#1}%
\providecommand \citenamefont [1]{#1}%
\providecommand \href@noop [0]{\@secondoftwo}%
\providecommand \href [0]{\begingroup \@sanitize@url \@href}%
\providecommand \@href[1]{\@@startlink{#1}\@@href}%
\providecommand \@@href[1]{\endgroup#1\@@endlink}%
\providecommand \@sanitize@url [0]{\catcode `\\12\catcode `\$12\catcode
  `\&12\catcode `\#12\catcode `\^12\catcode `\_12\catcode `\%12\relax}%
\providecommand \@@startlink[1]{}%
\providecommand \@@endlink[0]{}%
\providecommand \url  [0]{\begingroup\@sanitize@url \@url }%
\providecommand \@url [1]{\endgroup\@href {#1}{\urlprefix }}%
\providecommand \urlprefix  [0]{URL }%
\providecommand \Eprint [0]{\href }%
\providecommand \doibase [0]{https://doi.org/}%
\providecommand \selectlanguage [0]{\@gobble}%
\providecommand \bibinfo  [0]{\@secondoftwo}%
\providecommand \bibfield  [0]{\@secondoftwo}%
\providecommand \translation [1]{[#1]}%
\providecommand \BibitemOpen [0]{}%
\providecommand \bibitemStop [0]{}%
\providecommand \bibitemNoStop [0]{.\EOS\space}%
\providecommand \EOS [0]{\spacefactor3000\relax}%
\providecommand \BibitemShut  [1]{\csname bibitem#1\endcsname}%
\let\auto@bib@innerbib\@empty
\bibitem [{\citenamefont {Marconi}\ \emph {et~al.}(2026)\citenamefont
  {Marconi}, \citenamefont {Müller-Rigat}, \citenamefont {Romero-Pallejà},
  \citenamefont {Tura},\ and\ \citenamefont {Sanpera}}]{Marconi:2026}%
  \BibitemOpen
  \bibfield  {author} {\bibinfo {author} {\bibfnamefont {C.}~\bibnamefont
  {Marconi}}, \bibinfo {author} {\bibfnamefont {G.}~\bibnamefont
  {Müller-Rigat}}, \bibinfo {author} {\bibfnamefont {J.}~\bibnamefont
  {Romero-Pallejà}}, \bibinfo {author} {\bibfnamefont {J.}~\bibnamefont
  {Tura}},\ and\ \bibinfo {author} {\bibfnamefont {A.}~\bibnamefont
  {Sanpera}},\ }\bibfield  {title} {\bibinfo {title} {Symmetric quantum states:
  a review of recent progress},\ }\href
  {https://doi.org/10.1088/1361-6633/ae440a} {\bibfield  {journal} {\bibinfo
  {journal} {Rep. Prog. Phys.}\ }\textbf {\bibinfo {volume} {89}},\ \bibinfo
  {pages} {024001} (\bibinfo {year} {2026})}\BibitemShut {NoStop}%
\bibitem [{\citenamefont {Nielsen}\ and\ \citenamefont
  {Chuang}(2010)}]{Nielsen:2010}%
  \BibitemOpen
  \bibfield  {author} {\bibinfo {author} {\bibfnamefont {M.~A.}\ \bibnamefont
  {Nielsen}}\ and\ \bibinfo {author} {\bibfnamefont {I.~L.}\ \bibnamefont
  {Chuang}},\ }\href@noop {} {\emph {\bibinfo {title} {Quantum Computation and
  Quantum Information}}}\ (\bibinfo  {publisher} {Cambridge University Press},\
  \bibinfo {year} {2010})\BibitemShut {NoStop}%
\bibitem [{\citenamefont {von Neumann}(1955)}]{vonNeumann:1955}%
  \BibitemOpen
  \bibfield  {author} {\bibinfo {author} {\bibfnamefont {J.}~\bibnamefont {von
  Neumann}},\ }\href@noop {} {\emph {\bibinfo {title} {Mathematical Foundations
  of Quantum Mechanics}}}\ (\bibinfo  {publisher} {Princeton University
  Press},\ \bibinfo {address} {Princeton, New Jersey},\ \bibinfo {year}
  {1955})\ \bibinfo {note} {english translation of the 1932 German edition
  \textit{Mathematische Grundlagen der Quantenmechanik}}\BibitemShut {NoStop}%
\bibitem [{\citenamefont {Hein}\ \emph {et~al.}(2004)\citenamefont {Hein},
  \citenamefont {Eisert},\ and\ \citenamefont {Briegel}}]{Hein:2004}%
  \BibitemOpen
  \bibfield  {author} {\bibinfo {author} {\bibfnamefont {M.}~\bibnamefont
  {Hein}}, \bibinfo {author} {\bibfnamefont {J.}~\bibnamefont {Eisert}},\ and\
  \bibinfo {author} {\bibfnamefont {H.~J.}\ \bibnamefont {Briegel}},\
  }\bibfield  {title} {\bibinfo {title} {Multiparty entanglement in graph
  states},\ }\href {https://doi.org/10.1103/PhysRevA.69.062311} {\bibfield
  {journal} {\bibinfo  {journal} {Phys. Rev. A}\ }\textbf {\bibinfo {volume}
  {69}},\ \bibinfo {pages} {062311} (\bibinfo {year} {2004})}\BibitemShut
  {NoStop}%
\bibitem [{\citenamefont {Aschauer}\ \emph {et~al.}(2005)\citenamefont
  {Aschauer}, \citenamefont {Dür},\ and\ \citenamefont
  {Briegel}}]{Aschauer:2005}%
  \BibitemOpen
  \bibfield  {author} {\bibinfo {author} {\bibfnamefont {H.}~\bibnamefont
  {Aschauer}}, \bibinfo {author} {\bibfnamefont {W.}~\bibnamefont {Dür}},\
  and\ \bibinfo {author} {\bibfnamefont {H.-J.}\ \bibnamefont {Briegel}},\
  }\bibfield  {title} {\bibinfo {title} {Multiparticle entanglement
  purification for two-colorable graph states},\ }\href
  {https://doi.org/10.1103/physreva.71.012319} {\bibfield  {journal} {\bibinfo
  {journal} {Phys. Rev. A}\ }\textbf {\bibinfo {volume} {71}},\ \bibinfo
  {pages} {012319} (\bibinfo {year} {2005})}\BibitemShut {NoStop}%
\bibitem [{\citenamefont {Raussendorf}\ and\ \citenamefont
  {Briegel}(2001)}]{Raussendorf:2001}%
  \BibitemOpen
  \bibfield  {author} {\bibinfo {author} {\bibfnamefont {R.}~\bibnamefont
  {Raussendorf}}\ and\ \bibinfo {author} {\bibfnamefont {H.~J.}\ \bibnamefont
  {Briegel}},\ }\bibfield  {title} {\bibinfo {title} {A one-way quantum
  computer},\ }\href {https://doi.org/10.1103/PhysRevLett.86.5188} {\bibfield
  {journal} {\bibinfo  {journal} {Phys. Rev. Lett.}\ }\textbf {\bibinfo
  {volume} {86}},\ \bibinfo {pages} {5188} (\bibinfo {year}
  {2001})}\BibitemShut {NoStop}%
\bibitem [{\citenamefont {Raussendorf}\ \emph {et~al.}(2003)\citenamefont
  {Raussendorf}, \citenamefont {Browne},\ and\ \citenamefont
  {Briegel}}]{Raussendorf:2003}%
  \BibitemOpen
  \bibfield  {author} {\bibinfo {author} {\bibfnamefont {R.}~\bibnamefont
  {Raussendorf}}, \bibinfo {author} {\bibfnamefont {D.~E.}\ \bibnamefont
  {Browne}},\ and\ \bibinfo {author} {\bibfnamefont {H.~J.}\ \bibnamefont
  {Briegel}},\ }\bibfield  {title} {\bibinfo {title} {Measurement-based quantum
  computation on cluster states},\ }\href
  {https://doi.org/10.1103/PhysRevA.68.022312} {\bibfield  {journal} {\bibinfo
  {journal} {Phys. Rev. A}\ }\textbf {\bibinfo {volume} {68}},\ \bibinfo
  {pages} {022312} (\bibinfo {year} {2003})}\BibitemShut {NoStop}%
\bibitem [{\citenamefont {Briegel}\ \emph {et~al.}(2009)\citenamefont
  {Briegel}, \citenamefont {Browne}, \citenamefont {D{\"u}r}, \citenamefont
  {Raussendorf},\ and\ \citenamefont {Van~den Nest}}]{Briegel:2009}%
  \BibitemOpen
  \bibfield  {author} {\bibinfo {author} {\bibfnamefont {H.~J.}\ \bibnamefont
  {Briegel}}, \bibinfo {author} {\bibfnamefont {D.~E.}\ \bibnamefont {Browne}},
  \bibinfo {author} {\bibfnamefont {W.}~\bibnamefont {D{\"u}r}}, \bibinfo
  {author} {\bibfnamefont {R.}~\bibnamefont {Raussendorf}},\ and\ \bibinfo
  {author} {\bibfnamefont {M.}~\bibnamefont {Van~den Nest}},\ }\bibfield
  {title} {\bibinfo {title} {Measurement-based quantum computation},\ }\href
  {https://doi.org/10.1038/nphys1157} {\bibfield  {journal} {\bibinfo
  {journal} {Nat. Phys.}\ }\textbf {\bibinfo {volume} {5}},\ \bibinfo {pages}
  {19} (\bibinfo {year} {2009})}\BibitemShut {NoStop}%
\bibitem [{\citenamefont {Ionicioiu}\ and\ \citenamefont
  {Spiller}(2012)}]{Ionicioiu:2012}%
  \BibitemOpen
  \bibfield  {author} {\bibinfo {author} {\bibfnamefont {R.}~\bibnamefont
  {Ionicioiu}}\ and\ \bibinfo {author} {\bibfnamefont {T.~P.}\ \bibnamefont
  {Spiller}},\ }\bibfield  {title} {\bibinfo {title} {Encoding graphs into
  quantum states: An axiomatic approach},\ }\href
  {https://doi.org/10.1103/PhysRevA.85.062313} {\bibfield  {journal} {\bibinfo
  {journal} {Phys. Rev. A}\ }\textbf {\bibinfo {volume} {85}},\ \bibinfo
  {pages} {062313} (\bibinfo {year} {2012})}\BibitemShut {NoStop}%
\bibitem [{\citenamefont {Qu}\ \emph {et~al.}(2013{\natexlab{a}})\citenamefont
  {Qu}, \citenamefont {Wang}, \citenamefont {Li},\ and\ \citenamefont
  {Bao}}]{Qu:2013a}%
  \BibitemOpen
  \bibfield  {author} {\bibinfo {author} {\bibfnamefont {R.}~\bibnamefont
  {Qu}}, \bibinfo {author} {\bibfnamefont {J.}~\bibnamefont {Wang}}, \bibinfo
  {author} {\bibfnamefont {Z.-S.}\ \bibnamefont {Li}},\ and\ \bibinfo {author}
  {\bibfnamefont {Y.-R.}\ \bibnamefont {Bao}},\ }\bibfield  {title} {\bibinfo
  {title} {Encoding hypergraphs into quantum states},\ }\href
  {https://doi.org/10.1103/PhysRevA.87.022311} {\bibfield  {journal} {\bibinfo
  {journal} {Phys. Rev. A}\ }\textbf {\bibinfo {volume} {87}},\ \bibinfo
  {pages} {022311} (\bibinfo {year} {2013}{\natexlab{a}})}\BibitemShut
  {NoStop}%
\bibitem [{\citenamefont {Hein}\ \emph {et~al.}(2006)\citenamefont {Hein},
  \citenamefont {D{\"u}r}, \citenamefont {Eisert}, \citenamefont {Raussendorf},
  \citenamefont {Van~den Nest},\ and\ \citenamefont {Briegel}}]{Hein:2006}%
  \BibitemOpen
  \bibfield  {author} {\bibinfo {author} {\bibfnamefont {M.}~\bibnamefont
  {Hein}}, \bibinfo {author} {\bibfnamefont {W.}~\bibnamefont {D{\"u}r}},
  \bibinfo {author} {\bibfnamefont {J.}~\bibnamefont {Eisert}}, \bibinfo
  {author} {\bibfnamefont {R.}~\bibnamefont {Raussendorf}}, \bibinfo {author}
  {\bibfnamefont {M.}~\bibnamefont {Van~den Nest}},\ and\ \bibinfo {author}
  {\bibfnamefont {H.~J.}\ \bibnamefont {Briegel}},\ }\bibfield  {title}
  {\bibinfo {title} {Entanglement in graph states and its applications},\
  }\href {https://doi.org/10.3254/978-1-61499-018-5-115} {\bibfield  {journal}
  {\bibinfo  {journal} {Proceedings of the International School of Physics
  "Enrico Fermi"}\ }\textbf {\bibinfo {volume} {162}},\ \bibinfo {pages} {115}
  (\bibinfo {year} {2006})}\BibitemShut {NoStop}%
\bibitem [{\citenamefont {Qu}\ \emph {et~al.}(2013{\natexlab{b}})\citenamefont
  {Qu}, \citenamefont {Li}, \citenamefont {Wang},\ and\ \citenamefont
  {Bao}}]{Qu:2013b}%
  \BibitemOpen
  \bibfield  {author} {\bibinfo {author} {\bibfnamefont {R.}~\bibnamefont
  {Qu}}, \bibinfo {author} {\bibfnamefont {Z.-s.}\ \bibnamefont {Li}}, \bibinfo
  {author} {\bibfnamefont {J.}~\bibnamefont {Wang}},\ and\ \bibinfo {author}
  {\bibfnamefont {Y.-r.}\ \bibnamefont {Bao}},\ }\bibfield  {title} {\bibinfo
  {title} {Multipartite entanglement and hypergraph states of three qubits},\
  }\href {https://doi.org/10.1103/physreva.87.032329} {\bibfield  {journal}
  {\bibinfo  {journal} {Phys. Rev. A}\ }\textbf {\bibinfo {volume} {87}},\
  \bibinfo {pages} {032329} (\bibinfo {year} {2013}{\natexlab{b}})}\BibitemShut
  {NoStop}%
\bibitem [{\citenamefont {Qu}\ \emph {et~al.}(2013{\natexlab{c}})\citenamefont
  {Qu}, \citenamefont {Ma}, \citenamefont {Wang},\ and\ \citenamefont
  {Bao}}]{Qu:2013c}%
  \BibitemOpen
  \bibfield  {author} {\bibinfo {author} {\bibfnamefont {R.}~\bibnamefont
  {Qu}}, \bibinfo {author} {\bibfnamefont {Y.-p.}\ \bibnamefont {Ma}}, \bibinfo
  {author} {\bibfnamefont {B.}~\bibnamefont {Wang}},\ and\ \bibinfo {author}
  {\bibfnamefont {Y.-r.}\ \bibnamefont {Bao}},\ }\bibfield  {title} {\bibinfo
  {title} {Relationship among locally maximally entangleable states,wstates,
  and hypergraph states under local unitary transformations},\ }\href
  {https://doi.org/10.1103/physreva.87.052331} {\bibfield  {journal} {\bibinfo
  {journal} {Phys. Rev. A}\ }\textbf {\bibinfo {volume} {87}},\ \bibinfo
  {pages} {052331} (\bibinfo {year} {2013}{\natexlab{c}})}\BibitemShut
  {NoStop}%
\bibitem [{\citenamefont {Rossi}\ \emph {et~al.}(2013)\citenamefont {Rossi},
  \citenamefont {Huber}, \citenamefont {Bruß},\ and\ \citenamefont
  {Macchiavello}}]{Rossi:2013}%
  \BibitemOpen
  \bibfield  {author} {\bibinfo {author} {\bibfnamefont {M.}~\bibnamefont
  {Rossi}}, \bibinfo {author} {\bibfnamefont {M.}~\bibnamefont {Huber}},
  \bibinfo {author} {\bibfnamefont {D.}~\bibnamefont {Bruß}},\ and\ \bibinfo
  {author} {\bibfnamefont {C.}~\bibnamefont {Macchiavello}},\ }\bibfield
  {title} {\bibinfo {title} {Quantum hypergraph states},\ }\href
  {https://doi.org/10.1088/1367-2630/15/11/113022} {\bibfield  {journal}
  {\bibinfo  {journal} {New J. Phys.}\ }\textbf {\bibinfo {volume} {15}},\
  \bibinfo {pages} {113022} (\bibinfo {year} {2013})}\BibitemShut {NoStop}%
\bibitem [{\citenamefont {G{\"u}hne}\ \emph {et~al.}(2014)\citenamefont
  {G{\"u}hne}, \citenamefont {Cuquet}, \citenamefont {Steinhoff}, \citenamefont
  {Moroder}, \citenamefont {Rossi}, \citenamefont {Bru{\ss}}, \citenamefont
  {Kraus},\ and\ \citenamefont {Macchiavello}}]{Guhne:2014}%
  \BibitemOpen
  \bibfield  {author} {\bibinfo {author} {\bibfnamefont {O.}~\bibnamefont
  {G{\"u}hne}}, \bibinfo {author} {\bibfnamefont {M.}~\bibnamefont {Cuquet}},
  \bibinfo {author} {\bibfnamefont {F.~E.~S.}\ \bibnamefont {Steinhoff}},
  \bibinfo {author} {\bibfnamefont {T.}~\bibnamefont {Moroder}}, \bibinfo
  {author} {\bibfnamefont {M.}~\bibnamefont {Rossi}}, \bibinfo {author}
  {\bibfnamefont {D.}~\bibnamefont {Bru{\ss}}}, \bibinfo {author}
  {\bibfnamefont {B.}~\bibnamefont {Kraus}},\ and\ \bibinfo {author}
  {\bibfnamefont {C.}~\bibnamefont {Macchiavello}},\ }\bibfield  {title}
  {\bibinfo {title} {Entanglement and nonclassical properties of hypergraph
  states},\ }\href {https://doi.org/10.1088/1751-8113/47/33/335303} {\bibfield
  {journal} {\bibinfo  {journal} {J. Phys. A}\ }\textbf {\bibinfo {volume}
  {47}},\ \bibinfo {pages} {335303} (\bibinfo {year} {2014})}\BibitemShut
  {NoStop}%
\bibitem [{\citenamefont {Poderini}\ \emph {et~al.}(2026)\citenamefont
  {Poderini}, \citenamefont {Bruß},\ and\ \citenamefont
  {Macchiavello}}]{Poderini:2026}%
  \BibitemOpen
  \bibfield  {author} {\bibinfo {author} {\bibfnamefont {D.}~\bibnamefont
  {Poderini}}, \bibinfo {author} {\bibfnamefont {D.}~\bibnamefont {Bruß}},\
  and\ \bibinfo {author} {\bibfnamefont {C.}~\bibnamefont {Macchiavello}},\
  }\bibfield  {title} {\bibinfo {title} {Quantum hypergraph states: a review},\
  }\href {https://doi.org/10.1088/1361-6633/ae7735} {\bibfield  {journal}
  {\bibinfo  {journal} {Rep. Prog. Phys.}\ }\textbf {\bibinfo {volume} {89}},\
  \bibinfo {pages} {066001} (\bibinfo {year} {2026})}\BibitemShut {NoStop}%
\bibitem [{\citenamefont {Oxley}(1992)}]{Oxley:1992}%
  \BibitemOpen
  \bibfield  {author} {\bibinfo {author} {\bibfnamefont {J.~G.}\ \bibnamefont
  {Oxley}},\ }\href@noop {} {\emph {\bibinfo {title} {Matroid Theory}}}\
  (\bibinfo  {publisher} {Oxford University Press},\ \bibinfo {address} {New
  York},\ \bibinfo {year} {1992})\BibitemShut {NoStop}%
\bibitem [{\citenamefont {Whitney}(1935)}]{Whitney:1935}%
  \BibitemOpen
  \bibfield  {author} {\bibinfo {author} {\bibfnamefont {H.}~\bibnamefont
  {Whitney}},\ }\bibfield  {title} {\bibinfo {title} {On the abstract
  properties of linear dependence},\ }\href {https://doi.org/10.2307/2371182}
  {\bibfield  {journal} {\bibinfo  {journal} {Am. J. Math.}\ }\textbf {\bibinfo
  {volume} {57}},\ \bibinfo {pages} {509} (\bibinfo {year} {1935})}\BibitemShut
  {NoStop}%
\bibitem [{\citenamefont {Liu}\ and\ \citenamefont {Winter}(2022)}]{Liu:2022}%
  \BibitemOpen
  \bibfield  {author} {\bibinfo {author} {\bibfnamefont {Z.-W.}\ \bibnamefont
  {Liu}}\ and\ \bibinfo {author} {\bibfnamefont {A.}~\bibnamefont {Winter}},\
  }\bibfield  {title} {\bibinfo {title} {Many-body quantum magic},\ }\href
  {https://doi.org/10.1103/prxquantum.3.020333} {\bibfield  {journal} {\bibinfo
   {journal} {PRX Quantum}\ }\textbf {\bibinfo {volume} {3}},\ \bibinfo {pages}
  {020333} (\bibinfo {year} {2022})}\BibitemShut {NoStop}%
\bibitem [{\citenamefont {Gottesman}(1996)}]{Gottesman:1996}%
  \BibitemOpen
  \bibfield  {author} {\bibinfo {author} {\bibfnamefont {D.}~\bibnamefont
  {Gottesman}},\ }\bibfield  {title} {\bibinfo {title} {Class of quantum
  error-correcting codes saturating the quantum hamming bound},\ }\href
  {https://doi.org/10.1103/PhysRevA.54.1862} {\bibfield  {journal} {\bibinfo
  {journal} {Phys. Rev. A}\ }\textbf {\bibinfo {volume} {54}},\ \bibinfo
  {pages} {1862} (\bibinfo {year} {1996})}\BibitemShut {NoStop}%
\bibitem [{\citenamefont {Gottesman}(1998)}]{Gottesman:1998}%
  \BibitemOpen
  \bibfield  {author} {\bibinfo {author} {\bibfnamefont {D.}~\bibnamefont
  {Gottesman}},\ }\bibfield  {title} {\bibinfo {title} {Theory of
  fault-tolerant quantum computation},\ }\href
  {https://doi.org/10.1103/PhysRevA.57.127} {\bibfield  {journal} {\bibinfo
  {journal} {Phys. Rev. A}\ }\textbf {\bibinfo {volume} {57}},\ \bibinfo
  {pages} {127} (\bibinfo {year} {1998})}\BibitemShut {NoStop}%
\bibitem [{\citenamefont {Calderbank}\ \emph {et~al.}(1998)\citenamefont
  {Calderbank}, \citenamefont {Rains}, \citenamefont {Shor},\ and\
  \citenamefont {Sloane}}]{Calderbank:1998}%
  \BibitemOpen
  \bibfield  {author} {\bibinfo {author} {\bibfnamefont {A.~R.}\ \bibnamefont
  {Calderbank}}, \bibinfo {author} {\bibfnamefont {E.~M.}\ \bibnamefont
  {Rains}}, \bibinfo {author} {\bibfnamefont {P.~W.}\ \bibnamefont {Shor}},\
  and\ \bibinfo {author} {\bibfnamefont {N.~J.~A.}\ \bibnamefont {Sloane}},\
  }\bibfield  {title} {\bibinfo {title} {Quantum error correction via codes
  over gf(4)},\ }\href {https://doi.org/10.1109/18.681315} {\bibfield
  {journal} {\bibinfo  {journal} {IEEE Trans. Inf. Theory}\ }\textbf {\bibinfo
  {volume} {44}},\ \bibinfo {pages} {1369} (\bibinfo {year}
  {1998})}\BibitemShut {NoStop}%
\bibitem [{\citenamefont {Ketkar}\ \emph {et~al.}(2006)\citenamefont {Ketkar},
  \citenamefont {Klappenecker}, \citenamefont {Kumar},\ and\ \citenamefont
  {Sarvepalli}}]{Ketkar:2006}%
  \BibitemOpen
  \bibfield  {author} {\bibinfo {author} {\bibfnamefont {A.}~\bibnamefont
  {Ketkar}}, \bibinfo {author} {\bibfnamefont {A.}~\bibnamefont
  {Klappenecker}}, \bibinfo {author} {\bibfnamefont {S.}~\bibnamefont
  {Kumar}},\ and\ \bibinfo {author} {\bibfnamefont {P.~K.}\ \bibnamefont
  {Sarvepalli}},\ }\bibfield  {title} {\bibinfo {title} {Nonbinary stabilizer
  codes over finite fields},\ }\href {https://doi.org/10.1109/TIT.2006.883612}
  {\bibfield  {journal} {\bibinfo  {journal} {IEEE Trans. Inf. Theory}\
  }\textbf {\bibinfo {volume} {52}},\ \bibinfo {pages} {4892} (\bibinfo {year}
  {2006})}\BibitemShut {NoStop}%
\bibitem [{\citenamefont {Huang}\ \emph {et~al.}(2026)\citenamefont {Huang},
  \citenamefont {Shi}, \citenamefont {Zhang},\ and\ \citenamefont
  {Li}}]{Huang:2026}%
  \BibitemOpen
  \bibfield  {author} {\bibinfo {author} {\bibfnamefont {X.}~\bibnamefont
  {Huang}}, \bibinfo {author} {\bibfnamefont {F.}~\bibnamefont {Shi}}, \bibinfo
  {author} {\bibfnamefont {L.}~\bibnamefont {Zhang}},\ and\ \bibinfo {author}
  {\bibfnamefont {L.}~\bibnamefont {Li}},\ }\bibfield  {title} {\bibinfo
  {title} {Quantum states supported by matroids},\ }\href
  {https://doi.org/10.1103/mx1w-k8fg} {\bibfield  {journal} {\bibinfo
  {journal} {Phys. Rev. A}\ }\textbf {\bibinfo {volume} {113}},\ \bibinfo
  {pages} {062452} (\bibinfo {year} {2026})}\BibitemShut {NoStop}%
\end{thebibliography}%

\end{document}